\documentclass[11pt]{article}

\usepackage[T1]{fontenc}
\usepackage[utf8]{inputenc}
\usepackage[francais,english]{babel}
\usepackage{amsmath,latexsym,amsfonts,amsthm,bm,mathtools}
\usepackage{graphicx}
\usepackage[caption=false,font=footnotesize]{subfig}
\usepackage[top=1in, bottom=1in, left=1in, right=1in]{geometry}
\usepackage{setspace}
\usepackage{floatrow}
\usepackage{authblk}

\newtheorem{lemma}{Lemma}
\newtheorem{proposition}{Proposition}
\newtheorem{theorem}{Theorem}
\theoremstyle{definition}
\newtheorem{remark}{Remark}
\newtheorem{algorithm}{Algorithm}
\newcommand{\supp}{\text{supp}\,}
\newcommand{\R}{\mathbb R}
\newcommand{\E}{\mathbb E}
\newcommand{\T}{\text T}

\DeclareMathOperator*{\argmin}{arg\,min}
\providecommand{\keywords}[1]{\textbf{\textit{Keywords---}} #1}

\begin{document}

\title{Optimal Riemannian quantization\\
with an application to air traffic analysis}
\date{}
\author[1]{Alice Le Brigant}
\author[1]{Stéphane Puechmorel}
\affil[1]{ENAC, Université de Toulouse, Toulouse France}
\maketitle

\begin{abstract}
The goal of optimal quantization is to find the best approximation of a probability distribution by a discrete measure with finite support. When dealing with empirical distributions, this boils down to finding the best summary of the data by a smaller number of points, and automatically yields a $K$-means-type clustering. In this paper, we introduce Competitive Learning Riemannian Quantization (CLRQ), an online algorithm that computes the optimal summary when the data does not belong to a vector space, but rather a Riemannian manifold. We prove its convergence and show simulated examples on the sphere and the hyperbolic plane. We also provide an application to real data by using CLRQ to create summaries of images of covariance matrices estimated from air traffic images. These summaries are representative of the air traffic complexity and yield clusterings of the airspaces into zones that are homogeneous with respect to that criterion. They can then be compared using discrete optimal transport and be further used as inputs of a machine learning algorithm or as indexes in a traffic database.
\end{abstract}


\section{Introduction}

Most of the statistical tools developed so far are dedicated to data belonging to vector spaces, since it is the most convenient setting for algorithms as well as theoretical derivations. However, when dealing with real world applications, such a framework may not fit with the structure of the data. It is obviously the case for geostatistics over world-sized datasets, but it is also true in many other fields:  shapes in computer vision, diffusion tensor images in neuroimaging, signals in radar processing do not belong to a Euclidean space, but rather to a differentiable manifold. Riemannian geometry provides a convenient framework to deal with such objects. It allows a straightforward generalization of basic statistical notions such as means and medians \cite{frechet1948,karcher1977,arnaudon2013medians}, covariance \cite{pennec2006intrinsic}, Gaussian distributions \cite{said2017}, and of usual linear operations such as principal component analysis \cite{fletcher2004,sommer2010}. The use of these statistical tools has met a growing interest in various fields, including shape analysis \cite{kendall1984}, computational anatomy \cite{fletcher2004}, medical imaging \cite{fletcher2007}, probability theory \cite{bigot2017}, and Radar signal processing \cite{arnaudon2013riemannian,lebrigant2017}.

In air traffic management (ATM), a major concern is the ability to infer an estimation of the complexity as perceived by a controller from the knowledge of aircraft trajectories in a given airspace, as depicted in Figure \ref{fig:traffic}. Many interdependent factors are involved in the cognitive process of a human controller, making the problem extremely difficult to solve, if even possible. However, there is a consensus among the experts on the importance of traffic disorder. As detailed in \cite{alldata2018}, a good way to estimate traffic disorder is to assume the spatial distribution of the aircraft velocities to be Gaussian and use the covariance function as an indicator for local complexity. This model requires to estimate a covariance matrix at each sample point of the airspace, and therefore yields a mapping from the plane to the space of symmetric, positive definite (SPD) matrices. Such a mapping will be referred to as an (SPD) image in the sequel. Working directly with the SPD images is an extremely computationally expensive task, that is unrealistic in practice. Moreover, the information provided is highly redundant, making it cumbersome to use in statistical analysis. To cope with this problem, we seek to produce summaries of the images. To do so, we model an SPD image as a realization of a random field with values in the space of SPD matrices. The values collected at each sample point in the image define an empirical probability distribution, supported on the space of SPD matrices. We propose to produce a summary of this empirical distribution using optimal quantization. 

\begin{figure}
\subfloat{\includegraphics[width=0.3\textwidth,height=0.3\textwidth]{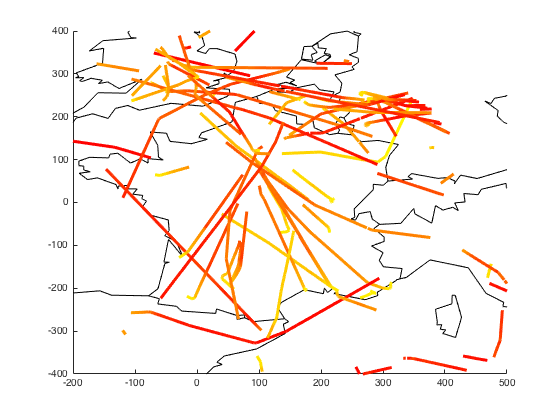}}
\subfloat{\includegraphics[width=0.3\textwidth,height=0.3\textwidth]{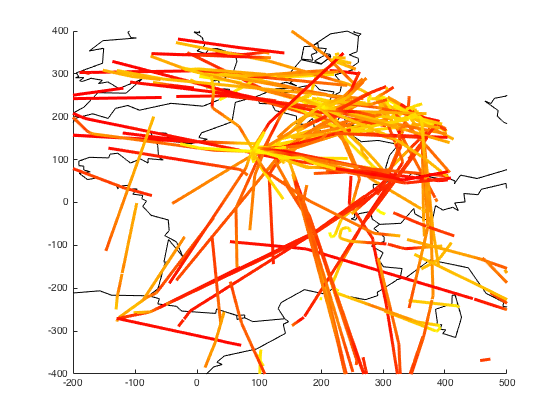}}
\subfloat{\includegraphics[width=0.3\textwidth,height=0.3\textwidth]{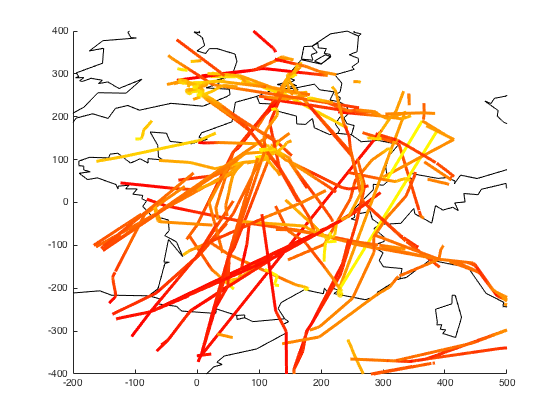}}
\caption{Traffic over the French airspace during one-hour periods of time.}
\label{fig:traffic}
\end{figure}

Optimal quantization is concerned with finding the best approximation, in the sense of the Wasserstein distance, of a probability distribution $\mu$ by a discrete measure $\hat\mu_n$ with finite support $|\text{supp}\hat\mu_n|\leq n$ (see \cite{graf2007} or the survey paper \cite{pages2015}, and references therin). When dealing with empirical distributions, this boils down to finding the best summary of the data by a smaller number of points. In the same setting, optimal quantization naturally yields a clustering of the data points, which coincides with the solution given by the $K$-means algorithm. In our application, probability measures are supported on the space of SPD matrices, and therefore points actually refer to SPD matrices. The main challenge lies in the fact that SPD matrices do not form a vector space, but rather a Riemannian manifold \cite{pennec2006riemannian}, whereas most work on optimal quantization is suited for vector data \cite{pages2004, bouchitte2011}, or functional data \cite{pages2006, biau2008}. Nonetheless, the case of probability distributions on Riemannian manifolds has recently received attention \cite{kloeckner2012,iacobelli2016}. In particular, the asymptotical behavior of the quantization error, i.e. the evolution of the error made by approximating $\mu$ by $\hat\mu_n$ as $n\to\infty$, was studied for the manifold case in \cite{iacobelli2016}. However, to the best of our knowledge, no numerical schemes have yet been introduced in this setting. 

In this work, we introduce Competitive Learning Riemannian Quantization (CLRQ), a Riemannian counterpart of Competitive Learning Vector Quantization (see for example \cite{pages2015}). It is a gradient descent algorithm that computes the best approximation $\hat\mu_n$ of a probability measure $\mu$ over a Riemannian manifold using observations sampled from $\mu$. For empirical distributions, this allows to summarize a manifold-valued dataset of size $N$ by a smaller number $n\ll N$ of points, which additionally correspond to the centers of a $K$-means-type clustering. Since it is an online algorithm, it is suited to process large datasets, and we show that it is convergent. We then use it to perform air traffic complexity analysis as presented above. Applying CLRQ to the empirical distribution of an SPD image yields two desirable results: (1) a summary of the image through $\hat\mu_n$, and (2) a clustering of the image, and thereby of the corresponding airspace, into different zones homogeneous with respect to complexity. The latter means that each point of the image is assigned a level of complexity according to the class it belongs to, and the former, that different traffic images can easily be compared through the comparison of their summaries, using the Wasserstein distance. This is an interesting prospect, since it allows for the indexation of air traffic situations in view of creating automatized decision-making tools to help air traffic controllers. Naturally, CLRQ can conceivably be applied to many other application-driven problems, requiring either the construction of summaries or the clustering of geometric data. A natural example of such data in positive curvature is given by GPS positions on the Earth (spherical data); in negative curvature, one can think of Gaussian distributions, which can be parameterized by the hyperbolic half-plane in the univariate case, and by higher dimensional Riemannian symmetric spaces in the multivariate case \cite{calvo1991,lovric2000}.

The paper is organized as follows. In Section \ref{sec:math}, we introduce the Riemannian setting and give an example of a statistical object on a manifold (the Riemannian center of mass), before introducing the context of optimal quantization. In Section \ref{sec:clqm}, we present the CLRQ algorithm and show its convergence. After showing some simple illustrations on the circle, the sphere and the hyperbolic half-plane in Section \ref{sec:ex}, we present our main application in air traffic management in Section \ref{sec:atm}.

\section{Mathematical setup}\label{sec:math}

\subsection{Notions of Riemannian geometry}

Let us begin by introducing some notations and reminding some basic notions of Riemannian geometry. We consider a differentiable manifold $M$ of dimension $d$ equipped with a Riemannian metric, i.e. a smoothly varying inner product $\langle\cdot,\cdot\rangle_x$ defined on each tangent space $T_xM$ at $x\in M$. Recall that $T_xM$ is a linear approximation of $M$ at point $x$, and contains all tangent vectors to $M$ at $x$. The subscript $x$ in the metric will often be omitted and the norm associated to the Riemannian metric $\langle\cdot,\cdot\rangle$ will be denoted by $\|\cdot\|$. Letting $\nabla$ be the Levi-Civita connection, the geodesics of $M$ are the curves $\gamma$ satisfying the relation $\nabla_{\dot{\gamma}}\dot{\gamma}=0$, which implies that their speed has constant norm $\|\dot\gamma(t)\|=cst$. They are also the local minimizers of the arc length functional $l$:
\[
l \colon \gamma \mapsto \int_0^1 \|\dot{\gamma}(t) \| dt
\]
where in the previous expression curves are assumed, without loss of generality, to be defined over the interval $[0,1]$.
The exponential map at point $x$ is the mapping, denoted by $\exp_x$, that maps a tangent vector $v$ of an open set $U\subset T_xM$ to the endpoint $\gamma(1)$ of the geodesic $\gamma:[0,1]\rightarrow M$ verifying $\gamma(0)=x$, $\dot \gamma(0)=v$,
\begin{equation*}
\exp_x(v) = \gamma(1). 
\end{equation*}
Intuitively, the exponential map moves the point $x$ along the geodesic starting from $x$ at speed $v$ and stops after covering the length $\|v\|$. Conversely, the inverse of the exponential map $\exp_x^{-1}(y)$ gives the vector that maps $x$ to $y$, and it will be denoted by the more intuitive notation $\overrightarrow{xy}$. We assume that $M$ is complete, i.e. that the exponential map at $x$ is defined on the whole tangent space $T_xM$. By the Hopf-Rinow theorem, we know that $M$ is also geodesically complete, that is, any two points $x,y\in M$ can be joined by a geodesic of shortest length. This minimal length defines the geodesic distance between $x$ and $y$, denoted in the sequel by $d(x,y)$. For further details, we refer the reader to a standard textbook such as \cite{jost2008}. 

\subsection{The Riemannian center of mass}

Before introducing optimal quantization, let us briefly give an example of generalization of a fundamental statistical notion to the manifold setting: the Riemannian center of mass of a probability distribution, also called the Fr\'{e}chet mean \cite{frechet1948}. Consider a probability measure $\mu$ on $M$ with density and a compact support $K=\text{supp}\mu$. As a compact set, $K$ is contained in a geodesic ball $B(a,R)=\{x\in M, d(a,x)<R\}$, and we assume that this ball is convex, i.e. that for all points $x,y\in \bar B(a,R)$, there exists a unique minimizing geodesic from $x$ to $y$ (sufficiently small geodesic balls are convex \cite[Theorem 5.14]{cheeger2008}). Let $X$ be a random variable of law $\mu$. The Fr\'{e}chet mean $\bar x$ of $X$ generalizes the minimum-square-error property of the Euclidean mean
\begin{equation}
\label{mean}
\bar x = \mathbb E_\mu[X]=\argmin_{a\in M}\int_M d(x,a)^2\mu(dx).
\end{equation}
It can be characterized as the point where the gradient of the above functional vanishes, which gives
\begin{equation}
\label{mean2}
\int_M \overrightarrow{\bar xx}\mu(dx) = 0.
\end{equation}
This is due to the fact that the gradient of the functional $f:a\mapsto d(x,a)^2$ is given by $\nabla_af  = -2 \overrightarrow{ax}$ (see Lemma \ref{lemgrad} in the Appendix), and to the assumptions we made on $M$ and $\supp\mu$. The same assumptions guarantee that $\bar x$ exists and is unique \cite{pennec2006intrinsic}. When $\mu$ is an empirical probability measure equally distributed on $n$ data points $x_1,\hdots,x_N$, the mean can be computed using a gradient descent algorithm, where the update is given by
\begin{equation*}
\bar x\leftarrow \exp_{\bar x}\left(\frac{1}{N}\sum_{i=1}^N \overrightarrow{\bar xx_i}\right).
\end{equation*}
This algorithm is often denoted by the Karcher flow.

\subsection{Optimal Riemannian quantization}
Optimal quantization addresses the problem of approximating a random variable $X$ of distribution $\mu$ by a simplified version $q(X)$ where $q: M\rightarrow \Gamma\subset M$ is a measurable function with an image $\Gamma$ of cardinal at most $n$. The function $q$ is called an \emph{$n$-quantizer} and is chosen to minimize the $L^p$ criteria
\begin{equation}\label{cost1}
\mathbb E_\mu \left[d(X,q(X))^p\right].
\end{equation}
Since $q$ takes only a finite number of values, the distribution of $q(X)$ will be a finite sum of point measures. 
It is well known that, since any $n$-quantizer $q$ of image $\Gamma \subset M$ verifies for all $x\in M$, $d(x,q(x))\geq \inf_{a\in\Gamma}d(x,a)$, with equality if and only if $q(x)=\argmin_{a\in\Gamma} d(x,a)$, the optimal quantizer is the projection to the nearest neighbor of $\Gamma$. Moreover, if $|\Gamma|<n$ and $|\supp\mu|\geq n$, one easily checks that $q$ can always be improved, in the sense of criteria \eqref{cost1}, by adding an element to its image. This means that an optimal $n$-quantizer has an image of exactly $|\Gamma|=n$ points. Therefore, the optimal $n$-quantizer for criteria \eqref{cost1} is of the form $q_\Gamma: M\rightarrow \Gamma=\{a_1,\hdots,a_n\}$, where the $a_i$'s are pairwise distinct, and
\begin{equation*}
q_\Gamma(\cdot) = \sum_{i=1}^n a_i \mathbf{1}_{C_i(\Gamma)}(\cdot), \quad \text{where} \quad C_i(\Gamma)=\{ x\in M, d(x,a_i) \leq d(x,a_j) \, \forall j\neq i\}.
\end{equation*}
The set $C_i(\Gamma)$ is the $i^{th}$ Voronoi cell associated to $\Gamma$ and the union of all these cells form the Voronoi diagram. The quantization problem is therefore equivalent to the approximation of $\supp\mu$ by an $n$-tuple $(a_1,\hdots,a_n)\in M^n$ minimizing the cost function $F_{n,p}: M^n \rightarrow \R_+$,
\begin{equation}\label{cost2}
F_{n,p}(a_1,\hdots,a_n)= \E_\mu\left(\min_{1\leq i \leq n} d(X,a_i)^p\right)=\int_M \min_{1\leq i\leq n}d(x,a_i)^p\mu(\mathrm dx).
\end{equation}
This cost function is obtained by evaluating \eqref{cost1} at $q=q_\Gamma$, and is called the \emph{distorsion function} of order $p$ for $\mu$. Notice that if we seek to approximate $\mu$ by a single point $a\in M$ (i.e. $n=1$) with respect to an $L^2$ criteria ($p=2$), we retrieve the definition \eqref{mean} of the Riemannian center of mass. 

Finally, there is a third way, in addition to \eqref{cost1} and \eqref{cost2}, of expressing the quantization problem: it is also equivalent to the approximation of the measure $\mu$ by the closest discrete measure $\hat\mu_n$ supported by $n$ points, with respect to the Wasserstein distance of order $p$
\begin{equation}\label{cost3}
W_p(\mu,\hat\mu_n)=\inf_{P}\int d(u,v)^pdP(u,v).
\end{equation}
Here the infimum is taken over all measures $P$ on $M\times M$ with marginals $\mu$ and $\hat\mu_n$. One can construct an optimal discrete approximation $\hat\mu_n$ (i.e. a minimizer of \eqref{cost3}) from an optimal $n$-tuple $(a_1,\hdots,a_n)$ (i.e. a minimizer of \eqref{cost2}), and vice-versa, using
\begin{equation}\label{mun}
\hat\mu_n = \sum_{i=1}^n \mu\left(C_i(\Gamma)\right)\delta_{a_i},
\end{equation}
and then we have
\begin{equation*}
F_{n,p}(a_1,\hdots,a_n) = W_p(\mu,\hat\mu_n).
\end{equation*}
This is well known for the vector case \cite{graf2007} and applies verbatim to measures on manifolds. In the sequel, we will focus on the second formulation \eqref{cost2} of the quantization problem.

The first question that arises is the existence of a minimizer of \eqref{cost2}. Since we have assumed that $\mu$ has compact support, this existence is easily obtained.
\begin{proposition}
Let $M$ be a complete Riemannian manifold and $\mu$ a probability distribution on $M$ with density and a compact support. Then the distorsion function $F_{n,p}$ is continuous and admits a minimizer.
\end{proposition}
\begin{proof}
Just as in the vector case \cite[Proposition 1.1]{pages2015}, for any $x\in M$, the function $M^{n}\rightarrow \R_+$, $\alpha=(a_1,\hdots,a_n)\mapsto \min_{1\leq i \leq n} d(x,a_i)$ is 1-Lipschitz for the distance $d'(\alpha,\beta):=\max_{1\leq i\leq n} d(a_i,b_i)$, where $\beta=(b_1,\hdots,b_n)$. Therefore it is continuous, and so is its $p^{th}$ power. Since $K=\supp\mu$ is compact, for all $\alpha\in M^n$ and all $\beta$ in a neighborhood $B(a_1,r_1)\times\hdots\times B(a_n,r_n)$ of $\alpha$, we have
\begin{equation*}
\forall x\in K,\quad \min_{1\leq i\leq n}d(x,b_i)^p\leq \min_{1\leq i\leq n}\left(\sup_{y\in K}d(y,a_i) + r_i \right)^p<\infty.
\end{equation*}
So by dominated convergence, $F_{n,p}$ is continuous. Recall that as a compact set, $K$ is contained in a geodesic ball $B(a,R)$. If $\alpha=(a_1,\hdots,a_n)\in M^n$ is such that $d(a,a_i)>2R$ for at least one $a_i$, then for all $x\in K$, $d(x,a_i)\geq d(a,a_i)-R>R$, and so the same $n$-tuple where $a$ replaces $a_i$ is a better candidate to minimize $F_{n,p}$. We can therefore limit our search to $\bar B(a,2R)$, which is a closed and bounded subset of the complete manifold $M$, and thus compact. The continuous function $F_{n,p}$ reaches a minimum on this compact, which is an absolute minimum.
\end{proof}
The elements of a minimizer $\alpha=(a_1,\hdots,a_n)$ are called \emph{optimal $n$-centers} of $\mu$. The minimizer $\alpha$ is in general not unique, first of all because any permutation of $\alpha$ is still a minimizer, and secondly because any symmetry of $\mu$, if it exists, will transform $\alpha$ into another minimizer of $F_{n,p}$. For example, any rotation of the optimal $n$-centers of the uniform distribution on the sphere conserves optimality. 

The second question that comes naturally is: how does the error one makes by approximating $\mu$ by $\hat\mu_n$ (as given by \eqref{mun}) evolve when the number $n$ of points grows ? The \emph{$n$-th quantization error} is defined by
\begin{equation*}
V_{n,p}(\mu) = \inf_{(a_1,\hdots,a_n)\in M^n}F_{n,p}(a_1,\hdots,a_n)=\inf_{(a_1,\hdots,a_n)\in M^n}\int_M \min_{1\leq i\leq n}d(x,a_i)^p\mu(\mathrm dx).
\end{equation*}
In the vector case, Zador's theorem \cite[Theorem 6.2]{graf2007} tells us that it decreases to zero as $n^{-p/d}$, and that the limit of $n^{p/d}V_{n,p}(\mu)$ is proportional to the \emph{$p^{th}$ quantization coefficient}, i.e. the limit (which is also an infimum) when $\mu$ is the uniform distribution on the unit square of $\R^d$
\begin{equation*}
Q_p([0,1]^d) =\lim_{n\geq1} n^{p/d}V_{n,p}\left(\mathcal U([0,1]^d)\right).
\end{equation*}
Moreover, when $\mu$ is absolutely continuous with density $h$, the asymptotic empirical distribution of the optimal $n$-centers is proportional to $h^{d/(d+p)}$. 

In the case of a Riemannian manifold $M$, the moment condition of the flat case generalizes to a condition involving the curvature of $M$. The following term measures the maximal variation of the exponential map at $x\in M$ when restricted to a $(d-1)$-dimensional sphere $S_\rho\subset T_{x}M$ of radius $\rho$
\begin{equation*}
A_{x}(\rho)= \sup_{v\in S_\rho, w\in T_vS_\rho, \|w\|=\rho} \left\| d_v\exp_x(w)\right\|.
\end{equation*}
The following generalization of Zador's theorem to Riemannian quantization was proposed by Iacobelli \cite{iacobelli2016}.
\begin{theorem}[{\cite[Theorem 1.4 and Corollary 1.5]{iacobelli2016}}]
Let $M$ be a complete Riemannian manifold without boundary, and let $\mu = h\,d\text{vol} + \mu_s$ be a probability measure on $M$, where $d\text{vol}$ denotes the Riemannian volume form and $\mu_s$ the singular part of $\mu$. Assume there exist $x_0\in M$ and $\delta>0$ such that
\begin{equation*}
\int_M d(x,x_0)^{p+\delta} d\mu(x) + \int_M A_{x_0}(d(x,x_0)^p) d\mu(x) <\infty.
\end{equation*}
Then
\begin{equation*}
\lim_{n\to\infty} n^{p/d}V_{n,p}(\mu) = Q_r\left([0,1]^d\right) \|h\|_{d/(d+p)},
\end{equation*}
where $\|\cdot\|_{r}$ denotes the $L^r$-norm. In addition, if $\mu_s=0$ and $(a_1,\hdots,a_n)$ are optimal $n$-centers, then
\begin{equation*}
\frac{1}{n}\sum_{i=1}^n\delta_{a_i} \overset{D}{\longrightarrow} \lambda h^{d/(d+p)}\mathrm dx \quad \text{as } n\to\infty,
\end{equation*}
where $\overset{D}{\rightarrow}$ denotes convergence in distribution and $\lambda$ is the appropriate normalizing constant.
\end{theorem}

In this work, we are interested in finding numerical schemes to compute the optimal $n$-centers $\alpha=(a_1,\hdots,a_n)$ in practice from potentially large sets of data. To do so, we will search for the critical points of the distorsion function.

\section{Competitive Learning Riemannian Quantization}\label{sec:clqm}

\subsection{Differentiability of the distorsion function}

We assume that the only knowledge that we have of the probability measure $\mu$ that we want to approximate is through an online sequence of i.i.d. observations $X_1, X_2, \hdots$ sampled from $\mu$. A classical algorithm used for quadratic ($p=2$) vector quantization is the \emph{Competitive Learning Vector Quantization} algorithm, a stochastic gradient descent method based on the differentiability of the distorsion function $F_{n,2}$. We propose here a natural extension of this method to our setting, i.e. a compactly-supported probability measure on a complete Riemannian manifold. It relies on the differentiability of the distorsion function.
\begin{proposition}
\label{propdiff}
Let $\alpha=(a_1,\hdots,a_n)\in M^n$ be an $n$-tuple of pairwise distinct components and $p>1$. Then $F_{n,p}$ is differentiable and its gradient in $\alpha$ is
\begin{equation*}
\nabla_\alpha F_{n,p} = \left( -p\int_{\mathring{C_i}(\alpha)} \|\overset{\longrightarrow}{a_ix}\|^{p-1}\frac{\overset{\longrightarrow}{a_ix}}{\|\overset{\longrightarrow}{a_ix}\|}\,\mu(\mathrm dx)\right)_{1\leq i\leq n}\in T_\alpha M^n,
\end{equation*}
where $\mathring{C_i}(\alpha)$ is the interior of the $i^{th}$ Voronoi cell of $\alpha$ and $\overset{\longrightarrow}{xy}:=\exp_x^{-1}(y)$ denotes the vector that sends $x$ on $y$ through the exponential map. In particular, the gradient of the quadratic distorsion function is given by
\begin{equation}\label{grad}
\nabla_\alpha F_{n,2} = \left( -2\int_{\mathring{C_i}} \overset{\longrightarrow}{a_ix}\,\mu(\mathrm dx)\right)_{1\leq i\leq n}=-2\left( \mathbb E_\mu \mathbf{1}_{\{X\in \mathring{C_i}\}} \overset{\longrightarrow}{a_iX}\right)_{1\leq i\leq n}.
\end{equation}
\end{proposition}
\begin{remark}\label{rem1}
The first observation we can make is that optimal $n$-centers are Riemannian centers of mass of their Voronoi cells, as characterized by \eqref{mean2}. Hence, the term $n$-centers is as appropriate in the Riemannian setting as it is in the vector case. More generally, for any value of $p$, each $a_i$, $i=1,\hdots,n$, is the $p$-mean of its Voronoi cell, i.e. the minimizer of
\begin{equation*}
a\mapsto \int_{\mathring{C_i}(\alpha)}d(x,a)^p\mu(\mathrm dx).
\end{equation*}
Therefore, the optimal $n$-centers are always contained in the compact support of $\mu$. 
\end{remark}
\begin{remark}\label{rem2}
The second observation is that the opposite direction of the gradient is, on average, given by the vectors inside the expectation. Competitive learning quantization consists in following this direction at each step $k$, that is, updating only the center $a_i$ corresponding to the Voronoi cell of the new observation $X_k$, in the direction of that new observation. In the Riemannian setting, instead of moving along straight lines, we simply follow geodesics using the exponential map. 
\end{remark}
\begin{proof}
Let $\alpha=(a_1,\hdots,a_n)\in M^n$ be an $n$-tuple of pairwise distinct components, $w=(w_1,\hdots,w_n)\in T_\alpha M^n$ a tangent vector to $\alpha$, and let $(-\epsilon,\epsilon)\ni t\mapsto \alpha(t)=(a_1(t),\hdots,a_n(t))$ be a variation of $\alpha$ in the direction given by $w$, i.e. such that $a_i(0)=a_i$ et $\dot{a_i}(0)=w_i$ for all $i=1,\hdots,n$. The functional $t\mapsto \min_{1\leq i\leq n}d(x,a_i(t))^p$ is differentiable for all $x\not\in \cup_{1\leq i\leq n}\partial C_i(\alpha(t))$, that is $\mu$-almost everywhere since $\mu(\partial C_i(\alpha(t)))=0$ for all $i$ (this is shown in the Appendix, Lemma \ref{lemmass}). Its derivative in $t=0$ is given by
\begin{align*}
\left.\frac{d}{dt}\right|_{t=0} \min_{1\leq i\leq n}d(x,a_i(t))^p &= \sum_{i=1}^n \mathbf{1}_{\{x\in \mathring{C_i}\}} \left.\frac{d}{dt}\right|_{t=0} d(x,a_i(t))^p\\
&=-\sum_{i=1}^n \mathbf{1}_{\{x\in \mathring{C_i}\}}  \frac{p}{2}\,\left[d(x,a_i(0))^2\right]^{p/2-1} 2\left\langle \overrightarrow{a_i(0)x}, \dot{a_i}(0)\right\rangle\\
&=-\sum_{i=1}^n \mathbf{1}_{\{x\in \mathring{C_i}\}}  p\|\overrightarrow{a_ix}\|^{p-2}\left\langle \overrightarrow{a_ix}, w_i\right\rangle.
\end{align*}
To go from the second to the third line, we have used the well known property that for any given $x$, the gradient of the function $f:a\mapsto d(x,a)^2$ is given by $\nabla_af  = -2 \exp_a^{-1}x= -2 \overrightarrow{ax}$ (see Lemma \ref{lemgrad} in the Appendix). We obtain by Cauchy-Schwarz, since $x\in B(a,R)$ and $a_i\in B(a,2R)$ (recall that any $n$-tuple containing a coordinate outside of $B(a,2R)$ is a worse candidate than the same $n$-tuple where $a$ replaces $a_i$),
\begin{equation*}
\left|\left.\frac{d}{dt}\right|_{t=0} \min_{1\leq i\leq n}d(x,a_i(t))^p \right| \leq p\sum_{i=1}^n \|\overrightarrow{a_ix}\|^{p-1}\|w_i\|\leq p(3R)^{p-1}\sum_{i=1}^n\|w_i\|.
\end{equation*}
Therefore, by dominated convergence, $t\mapsto \psi_{n,p}(\alpha(t))$ is differentiable and its differential in $\alpha$ at $w$ is given by
\begin{align*}
T_\alpha F_{n,p}(w) = \sum_{i=1}^n \left\langle -2\int \mathbf{1}_{\{x\in \mathring{C_i}\}} \|\overset{\longrightarrow}{a_ix}\|^{p-1}\frac{\overset{\longrightarrow}{a_ix}}{\|\overset{\longrightarrow}{a_ix}\|}\,\mu(\mathrm dx), w_i\right\rangle= \langle\langle \nabla_\alpha F_{n,p},w\rangle\rangle,
\end{align*}
where $\langle\langle\cdot,\cdot\rangle\rangle$ denotes the $L^2$ metric on $M^n$, which gives the desired result for the gradient.
\end{proof}

\subsection{The algorithm}

Assume that we have access to $N$ independant and identically distributed observations $X_1,\hdots,X_N$. We choose a sequence of positive steps $(\gamma_k)_{k\geq 1}\subset (0,1)$ verifying the usual conditions
\begin{equation}\label{steps}
\sum_{k\geq 1}\gamma_k =+\infty, \quad \sum_{k\geq 1}\gamma_k^2 <+\infty. 
\end{equation}
We propose the following algorithm.
\begin{algorithm}[Competitive Learning Riemannian Quantization]\label{clrq}
\leavevmode\par \noindent
Initialization: $\alpha(0)=(a_1(0),\hdots,a_n(0))$.\\
For $k=0,\hdots,N-1$,
\begin{enumerate}
\item find $i=\underset{j}{\text{argmin}}\, d(X_{k+1},a_j(k))$,
\item set $\alpha(k+1)=(a_1(k+1),\hdots,a_n(k+1)$ where 
\begin{align*}
a_i(k+1) &= \exp_{a_i(k)}\left(\gamma_{k+1} \overrightarrow{a_i(k)X_k}\right),\\
a_j(k+1) &= a_j(k) \quad \forall j\neq i.
\end{align*}
\end{enumerate}
\end{algorithm}
The steps are chosen to be in $(0,1)$ so that at each iteration, the center that is updated stays in the same Voronoi cell, guaranteeing that the centers stay pairwise distinct (if initially pairwise distinct).

Now we show the convergence of Algorithm \ref{clrq}, using a theorem from Bonnabel \cite{bonnabel2013}.
\begin{proposition}
Assume that the injectivity radius of $M$ is uniformly bounded from below by some $I>0$, and let $(\alpha(k))_{k\geq 0}$ be computed using Algorithm \ref{clrq} and samples from a compactly supported distribution $\mu$. Then $F_{n,2}(\alpha(k))$ converges a.s. and $\nabla_{\alpha(k)}F_{n,2} \to 0$ as $k\to\infty$ a.s.
\end{proposition}
\begin{proof}
The proof relies on \cite[Theorem 1]{bonnabel2013}, and therefore we adopt the same notations. The cost function $F_{n,2}$ and its gradient can be respectively expressed as expectations of the functions
\begin{gather*}
Q(x,\alpha) = \min_{1\leq i\leq n}d(x,a_i)^2,\\
H(x,\alpha) = \nabla_\alpha Q(x,\alpha)=\left(-2\,\mathbf 1_{\mathring C_i}(x) \overrightarrow{a_ix}\right)_{1\leq i\leq n}.
\end{gather*}
This is once again due to Lemma \ref{lemgrad}, given in the Appendix. As stated in Remark \ref{rem1}, the $n$-centers $\alpha(k)=(a_1(k),\hdots,a_n(k))$ are each the barycenter of their Voronoi cell and therefore always stay in the same compact ball $B(a,R)$ as the data $X_1,X_2,\hdots$. Since $\|\mathbf 1_{\mathring C_i(x)}\overrightarrow{yx}\|\leq 2R$ for all $y,x\in K\subset B(a,R)$ and $i=1,\hdots,n$, the gradient $H$ is bounded on $K\times K$. That is, all the assumptions of  \cite[Theorem 1]{bonnabel2013} are verified and Algorithm \ref{clrq} converges.
\end{proof}

\section{Examples}\label{sec:ex}

Now let us show some toy examples on manifolds of constant sectional curvature: the circle, the $2$-sphere and the hyperbolic plane. To start, we show optimal discrete approximations of the uniform and the von Mises distributions on the circle $S^1$ (Figure \ref{fig:S1}). The top row shows the initialization (left) and result (middle) of Competitive Learning Riemannian Quantization (CLRQ) of size $n=6$ performed on $N=1000$ observations sampled from the uniform distribution, while the bottom row shows the initialization and result of CLRQ of size $n=5$ performed on $N=1000$ observations sampled from the von Mises distribution centered in $0$ with concentration parameter $K=5$. The centers are initialized uniformly on $[0,2\pi]$ and $[-\pi,\pi]$ respectively, $S^1$ being identified with $[0,2\pi)$. In order to reduce dependency on the initialization, each step $k$ is repeated a certain number $m$ of times. In other words, the same step size is used several times. In the uniform case, $m=10$ is sufficient to obtain a visually satisfying result. For the von Mises distribution, we choose $m=50$.

On the right-hand side of Figure \ref{fig:S1}, we plot the evolution of the Wasserstein distance between the initial distribution and its quantized version. (Recall that the quantization cost function \eqref{cost3} involves the $L^2$-Wasserstein distance.) As shown in \cite{rabin2011}, the computation of the $L^p$-Wasserstein distance between two measures on the circle can be reduced to the same operation in the unit interval $[0,1]$ by "cutting" the circle at a certain point $s\in S^1$, i.e. by identifying it with a unit length fundamental domain for $\R/\mathbb Z$. However, when the two measures are not evenly distributed, the optimal cut is easier to find in the $L^1$ case, therefore we choose to merely compute the weaker $L^1$-Wasserstein (or Kantorovich-Rubinstein) distance using the algorithm introduced in \cite{cabrelli1995}. We plot the distance between the measure $\mu$ and its approximation at each step $k$ of Algorithm \ref{clrq}
\begin{equation*}
\nu(k) = \sum_{i=1}^n \mu(C_i(k)) \delta_{a_i(k)},
\end{equation*}
where $(a_1(k),\hdots,a_n(k))$ are the $n$-centers at step $k$, $C_1(k),\hdots,C_n(k)$ are the corresponding Voronoi cells, and $\delta_x$ is the Dirac distribution at $x\in M$. Assuming that $N\gg n$, we can approximate $\mu$ by the empirical measure of the observations $x_1,\hdots, x_N$
\begin{equation*}
\hat\mu = \frac{1}{N} \sum_{i=1}^N \delta_{x_k}.
\end{equation*}
In order to compare two discrete measures with the same number of points, we then identify $\hat\mu$ and $\nu(k)$ with the measures obtained on the reunion of their supports by completing with zero masses. For both the uniform and the von Mises examples, the Wasserstein distance decreases as expected.

\begin{figure}
\includegraphics[width=0.26\linewidth,height=0.22\linewidth]{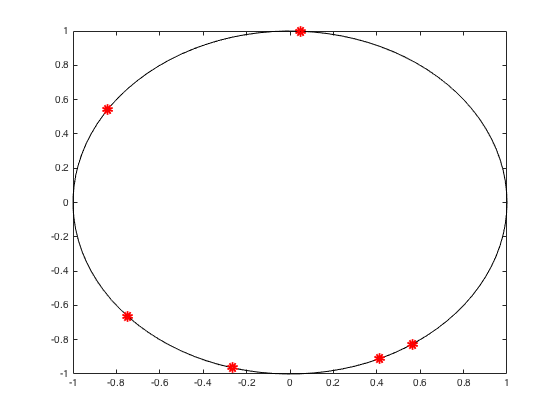}
\includegraphics[width=0.26\linewidth,height=0.22\linewidth]{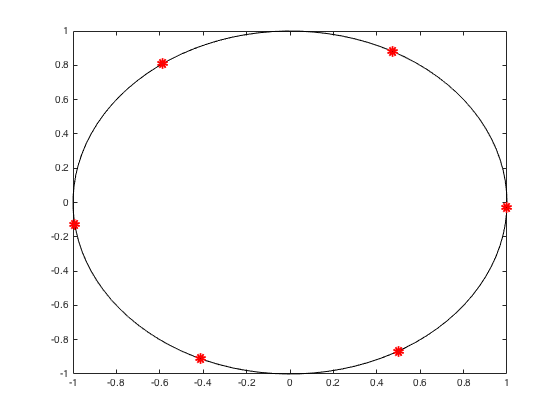}
\includegraphics[width=0.29\linewidth]{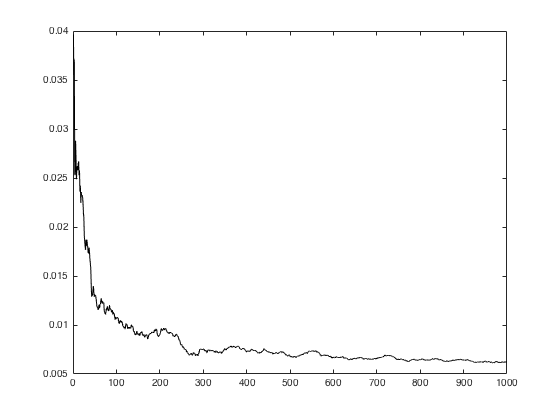}\\
\includegraphics[width=0.26\linewidth,height=0.22\linewidth]{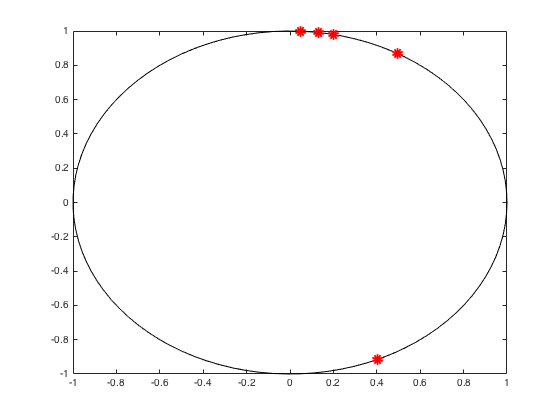}
\includegraphics[width=0.26\linewidth,height=0.22\linewidth]{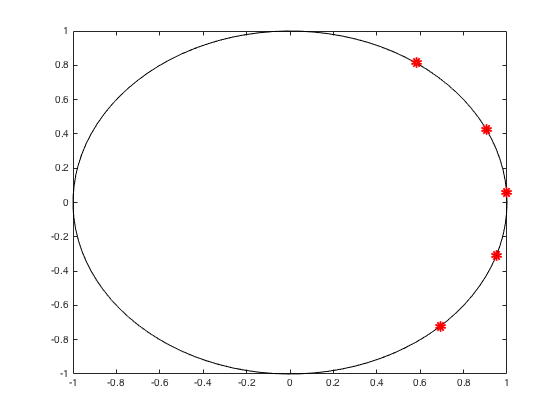}
\includegraphics[width=0.29\linewidth]{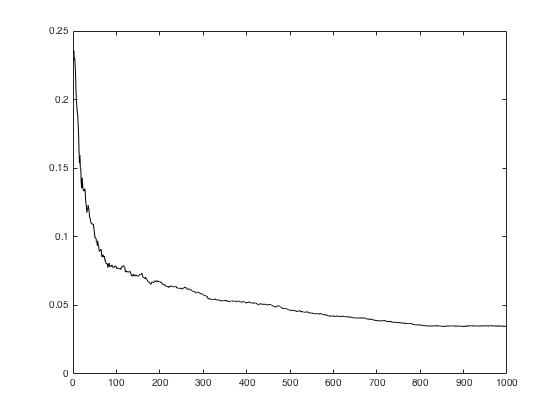}
\caption{Quantization of the uniform (top) and von Mises (bottom) distributions: initial positions of the $n$-centers (left), final positions of the $n$-centers (middle) and evolution of the $L^1$-Wasserstein distance between the initial distribution and its quantized version (right).}
\label{fig:S1}
\end{figure}

\begin{figure}
\subfloat{\includegraphics[width=0.275\textwidth,height=0.25\textwidth]{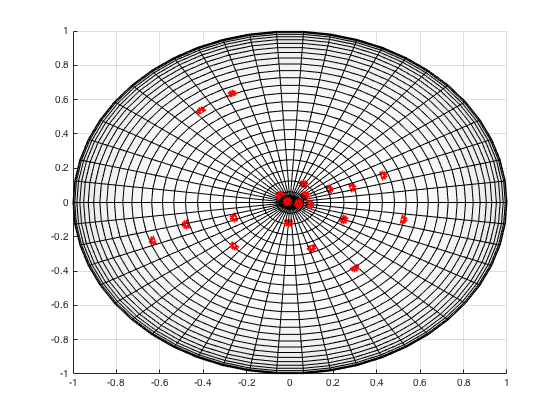}}
\subfloat{\includegraphics[width=0.275\textwidth,height=0.25\textwidth]{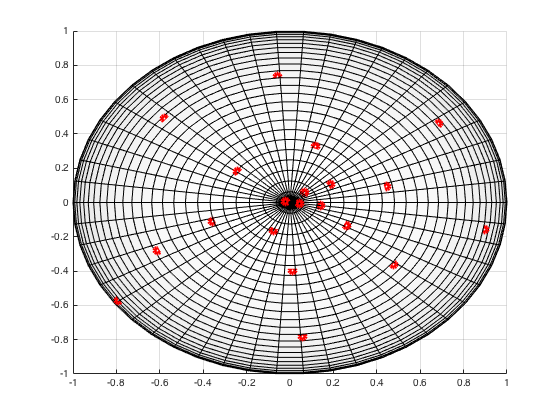}}
\subfloat{\includegraphics[width=0.275\textwidth,height=0.25\textwidth]{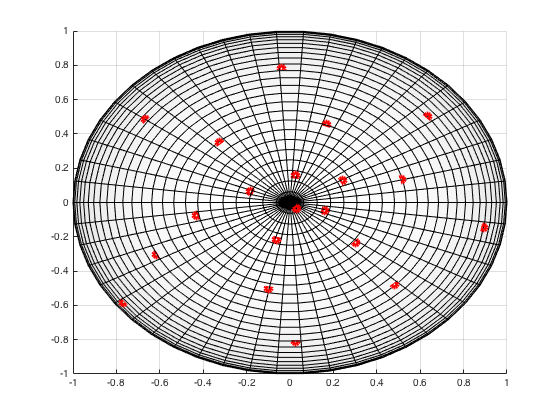}}\\
\caption{Quantization of the von Mises distribution on the $2$-sphere after 1 iteration (left), $50$ iterations (middle) and $200$ iterations (right).}
\label{fig:S2}
\end{figure}

\begin{figure}
\subfloat{\includegraphics[width=0.33\textwidth,height=0.3\textwidth]{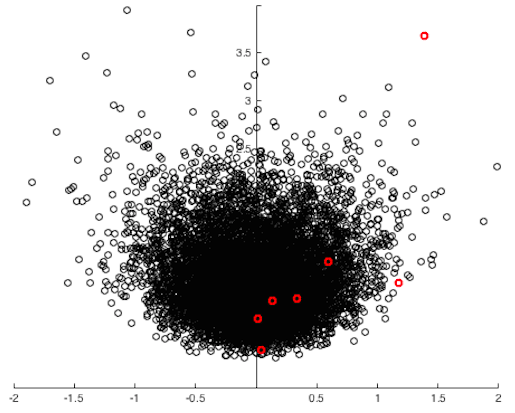}}
\subfloat{\includegraphics[width=0.33\textwidth,height=0.3\textwidth]{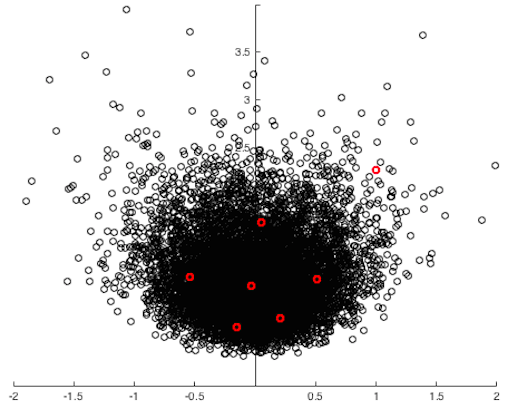}}
\subfloat{\includegraphics[width=0.33\textwidth,height=0.3\textwidth]{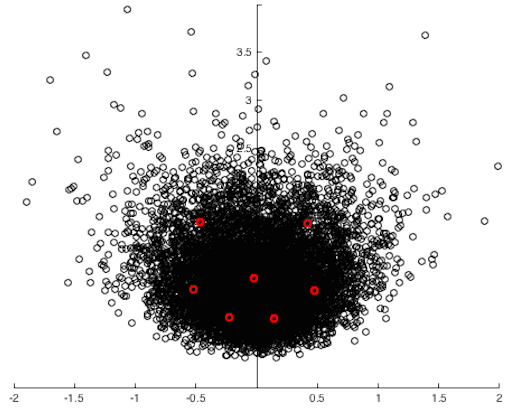}}
\caption{Quantization of the Gaussian distribution on the hyperbolic half-plane after 1, 20 and 100 iterations.}
\label{fig:H2}
\end{figure}

Next, we show examples on the sphere. Figure \ref{fig:S2} displays three steps of the CLRQ algorithm applied to the von Mises distribution with concentration parameter $K=5$, namely steps $1$, $50$ and $100$, where each step contains $m=50$ iterations at the same step size. 
Finally, to show an example in the negative curvature setting, we use the work of Said et al. \cite{said2017} regarding the definition and generation of Gaussian distributions on the space of SPD matrices to generate Gaussian samples on the hyperbolic half-plane. Recall that the hyperbolic half-plane is one of the models of $2$-dimensional hyperbolic geometry, and can be identified with the complex upper-half plane $\mathbb H^2 = \{z=x+iy, (x,y)\in \mathbb R\times \mathbb R_+^*\}$ where the length of an infinitesimal element $dx+idy$ at point $x+iy$ is measured by 
\begin{equation*}
ds^2=\frac{dx^2+dy^2}{y^2}.
\end{equation*}
The special linear group $\text{SL}_2$ acts on $\mathbb H^2$ from the left through the Moebius transformation: $\text{SO}_2\times\mathbb H^2\rightarrow \mathbb H^2$, defined by
\begin{equation*}
\left(\begin{matrix} a&b\\c&d\end{matrix}\right) \cdot z= \frac{az+b}{cz+d}.
\end{equation*}
This action is transitive since for all $(x,y)\in\mathbb R\times\mathbb R_+^*$,
\begin{equation*}
\left(\begin{matrix}y^{1/2} & xy^{-1/2}\\ 0 & y^{-1/2}\end{matrix}\right)\cdot i= x+iy.
\end{equation*}
Noticing that $\text{SO}_2$ is the stabilizer of $i$, we can identify $\mathbb H^2$ with $\text{SL}_2/\text{SO}_2$, which is also homeomorphic to the space of $2\times 2$ SPD matrices of determinant $1$ \cite{jost2008}. The space of SPD matrices  of determinant $1$ is therefore homeomorphic to $\mathbb H^2$, and the homeomorphism is given by $\Phi:P \mapsto L\cdot i$, where $L$ is the upper-triangular matrix of the Cholesky decomposition of the SPD matrix $P=L^TL$. To generate a Gaussian sample in $\mathbb H^2$, we generate a Gaussian sample of SPD matrices using \cite[Proposition 6]{said2017}, renormalize and transport them to the hyperbolic half-plane using $\Phi$. Figure \ref{fig:H2} shows steps $1$, $20$ and $100$ of the CLRQ algorithm applied to observations sampled from the Gaussian distribution centered in $i$ and with standard deviation $\sigma=0.5$. Each step contains $m=100$ iterations.

\section{Application to air traffic complexity analysis}\label{sec:atm}

\subsection{Context}

This work was motivated by complexity analysis in air traffic management (ATM). ATM deals with the means of organizing the aircraft trajectories in a given airspace so as to ensure both safety and efficiency. One of the most important part of ATM is the air traffic control (ATC) that acts on flying or taxiing aircrafts in such a way that separation norms are satisfied at all time. Nowadays, most of the ATC is surveillance based, relying primarily on  information coming from the Radars to give instructions to planes. Even in country-sized airspaces the amount of traffic to be controlled is far beyond the limits of a single operator and the area under the responsibility of an air traffic controller (ATCO) has to be kept tractable. As a consequence, the whole airspace must be divided into elementary cells, known as control sectors, that are alloted to a team of ATCOs. A major concern arising in the design of such sectors is to ensure that the ATCOs workload is equally balanced over the airspace. Highly complex areas, involving many flows crossings and altitude changes, like those encountered close to the largest airports, must be kept small, while low complexity ones, with either a small aircraft density or a simple airspace structure may be made large. Finding a complexity indicator that can be computed using only airspace and traffic information and that closely mimics the perceived workload is a difficult, still unsolved problem \cite{PRAN2011,COOK2015149}. One of the most widely used indicators is the dynamic density \cite{dyndens1998}, that combines influential factors , like number of maneuvering aircrafts, number of level changes and so on, to output a single positive real value representing the complexity level. Although quite pertinent from an operational point of view, the dynamic density is a tool requiring a lot of tuning, involving experiments with a panel of ATCOs and that cannot be adapted to different airspaces without having to re-tune from scratch. For the same reason, it is quite difficult to use it for assessing the performance of new concepts, since in such a case there is no reference situation or existing operational context that may be used to perform the tuning phase. 
On the other hand, purely geometrical indicators have been introduced \cite{lee2007,del2010}, that are free of the dynamic density limitations. While perfectly suited to infer an intrinsic notion of complexity, they do not model all the aspects of the workload, as perceived by a human operator. 
The approach taken in the present work may be viewed as a mix between the two previous ones: it relies on an intrinsic description of traffic, but does not produce a complexity value: instead, a summary of the situation is issued, that serves as an input to a classification or machine learning algorithm. Even if this last phase looks very similar to a kind of dynamic density evaluation, it is far less complex:
\begin{itemize}
\item The traffic summary itself requires no tuning, while the influential factors taken into account in the dynamic density have weights that must be adjusted during the learning phase. 
\item Complexity classes are computed instead of complexity levels: a clustering algorithm will first segment the traffic dataset into homogeneous subsets, then a workload value will be associated by experts to each of them. This process is lightweight, since only the representative in each class has to be expertized.
\item Adaptation to different airspaces is an easy task for the same reason: experts will evaluate only the representative situation in each class. 
\end{itemize}
The first step is to model the spatial distribution of the aircraft velocities as a Gaussian law. Then, the covariance function is used as an indicator of traffic complexity. This assumption makes sense from an operational point of view as it represents the local organization, that is the most influential factor on ATCCOs workload. 

\subsection{Estimating the covariance matrices}

Although we will in practice consider time-windows, we start by considering a given airspace at a fixed time $t$ containing $N$ aircrafts either flying or taxiing. We respectively denote by $z_i$ and $v_i$ the position and speed of the aircraft $i$, $1\leq i\leq N$, at time $t$. Since the altitude of an aircraft plays a special role and does not appear on controllers displays, we choose to disregard it and adopt a planar representation through a stereographic projection. 
An underlying Gaussian field model is assumed for the relation between the velocity and the position, whose variance function will be interpreted as a pointwise measure of the traffic complexity.

A non parametric approach of type Nadaraya-Watson \cite{nadaraya1964,watson1964} was taken to estimated the mean  and variance functions at point $z$:
\begin{equation}\label{cov}
\begin{aligned}
\hat m(z) &= \frac{\sum_{i=1}^NV_iK_h(z-Z_i)}{\sum_{j=1}^N K_h(z-Z_j)},\\
\hat\Sigma(z) &= \frac{\sum_{i=1}^N (V_i-\hat m(z))(V_i-\hat m(z))^T K_h(z-Z_i)}{\sum_{j=1}^N K_h(z-Z_j)}.
\end{aligned}
\end{equation}
The weights are given by a kernel function $K$, i.e. a positive, symmetric function of unit area, scaled by a factor $h>0$: $K_h(x)=h^{-1}K(x/h)$. Since most kernels have compact support, the estimations are based in practice on a number of observations that is very inferior to the size $N$ of the sample. The estimator $\hat\Sigma$ has been studied in \cite{yin2010} where it is shown to be asymptotically normal. Evaluating it at positions $z_1,\hdots,z_N$ yields a series of symmetric, positive definite matrices $\hat\Sigma(z_1),\hdots,\hat\Sigma(z_N)$ with empirical distribution
\begin{equation*}
\hat\mu=\frac{1}{N}\sum_{i=1}^N \delta_{\hat\Sigma(z_i)},
\end{equation*}
where $\delta_\Sigma$ denotes the Dirac mass at $\Sigma$. In order to obtain a summary of the traffic complexity, we propose to quantize $\hat\mu$ using the CLRQ algorithm on the space of SPD matrices.

\subsection{The geometry of SPD matrices}

For the sake of completeness, let us briefly recall the most commonly used Riemannian structure \cite{pennec2006riemannian} on the space $\mathcal P_n$ of symmetric, positive definite matrices. Note that in this application, we are simply interested in the case $n=2$. The Euclidean dot product on the space $\mathcal M_n$ of square matrices of size $n$ is given by the Frobenius inner product $\Sigma_1 \cdot \Sigma_2 = \text{tr}(\Sigma_1^\T\Sigma_2)$, where $\text{tr}$ denotes the trace. As an open subset of the vector space $\mathcal M_n$, $\mathcal P_n$ is a manifold where the tangent vectors are symmetric matrices. It can be equipped with a Riemannian metric invariant with respect to the action of the general linear group $GL_n\times\mathcal P_n\rightarrow \mathcal P_n$, $(A,\Sigma)\mapsto A^\T\Sigma A$. At the identity, this metric is given by the usual Euclidean scalar product $\langle W_1,W_2\rangle_{\text{Id}} = W_1\cdot W_2 = \text{tr}(W_1^\T W_2)$, and at $\Sigma$, we ask that the value of the scalar product does not change when the tangent vectors are transported back to the identity via the action of $A=\Sigma^{-1/2}$, i.e.
\begin{equation}
\langle W_1,W_2\rangle_{\Sigma} =\langle \Sigma^{-1/2}W_1\Sigma^{-1/2},\Sigma^{-1/2}W_2\Sigma^{-1/2}\rangle_{\text{Id}}= \text{tr}(\Sigma^{-1/2}W_1\Sigma^{-1}W_2\Sigma^{-1/2}).
\label{metric}
\end{equation}
The associated geodesic distance is given by
\begin{equation*}
d(\Sigma_1,\Sigma_2)=\sqrt{\sum_{i=1}^N \log^2\left(\lambda_i(\Sigma_1^{-1/2}\Sigma_2\Sigma_1^{-1/2})\right)},
\end{equation*}
where we use the notation $\lambda_i(\Sigma)$, $i=1,\hdots,n$, to denote the eigenvalues of $\Sigma$. Recall that in order to update the centers of the discrete approximation in the CLRQ algorithm, we need the exponential map, i.e. a mapping that associates to each point $\Sigma$ and tangent vector $W$ at $\Sigma$ the end point of the geodesic starting from $\Sigma$ at speed $W$. In the case of metric \eqref{metric}, it is given by
\begin{equation*}
\text{exp}_\Sigma(W)=\Sigma^{1/2}\text{exp}\left(\Sigma^{-1/2}W\Sigma^{-1/2}\right)\Sigma^{1/2},
\end{equation*}
where the $\text{exp}$ on the right-hand side denotes the matrix exponential. Finally, we also need the inverse mapping, i.e. the logarithm map
\begin{equation*}
\text{log}_{\Sigma_1}(\Sigma_2) = \overrightarrow{\Sigma_1\Sigma_2} = \Sigma_1^{1/2}\text{log}\left(\Sigma_1^{-1/2}\Sigma_2\Sigma_1^{-1/2}\right)\Sigma_1^{1/2},
\end{equation*}
where the $\text{log}$ on the right-hand side denotes the matrix logarithm. Note that the matrix logarithm is well defined for any symmetric matrix $\Sigma$ due to the factoring out in the logarithm series of the rotation matrices of the spectral decomposition $\Sigma=UDU^\T$.

\subsection{Real data analysis}

\subsubsection{Segmenting and constructing summaries}

We now have all the tools to construct summaries of the traffic complexity in a given airspace during a certain time period. As input, we consider an image such as the ones displayed in the first row of Figure \ref{fig:villes}, showing the traffic over Paris, Toulouse and Lyon during a one-hour period of time. The color is related to the norm of the velocity, increasing from yellow to red. To simplify, we center and reduce the velocities $v_i$. The samples $(z_i,v_i)$ are seen as observations arriving in a random order, and the covariance matrix at $z_i$ is estimated according to \eqref{cov} using a truncated Gaussian kernel $K(x)=1/\sqrt{2\pi} \,\,e^{-x^2/2}\mathbf 1_{|x|<r}$. The truncation of size $r$ avoids useless computations. 
\begin{algorithm}[CLRQ for air traffic complexity analysis]\label{clrqatm}
\leavevmode\par \noindent
Initialization: 
\begin{itemize}
\item Choose $i_1,\hdots,i_n$ randomly among $1,\hdots,N$
\item Compute $\hat\Sigma(z_{i_1}),\hdots, \hat\Sigma(z_{i_n})$ using \eqref{cov}. 
\item Set $A_1(0)=\hat\Sigma(z_{i_1}),\hdots,A_n(0)=\hat\Sigma(z_{i_n})$.
\end{itemize}
For $k=0,\hdots,N-1$, let $(z_{i_k},v_{i_k})$ be the current observation.
\begin{enumerate}
\item Compute $\hat \Sigma(z_{i_k})$ using \eqref{cov}.
\item Find $j_{min}=\underset{j}{\text{argmin}}\, d(\hat\Sigma(z_{i_k}),A_j(k))$.
\item Set $A(k+1)=(A_1(k+1),\hdots,A_n(k+1)$ where 
\begin{align*}
A_{j_{min}}(k+1) &= \exp_{A_{j_{min}}(k)}\left(\gamma_{k+1} \overrightarrow{A_{j_{min}}(k)\hat\Sigma(z_{i_k})}\right),\\
A_j(k+1) &= A_j(k) \quad \forall j\neq j_{min}.
\end{align*}
\end{enumerate}
\end{algorithm}
In practice, we usually look for $n=3$ centers, i.e. the best approximation by three points. Indeed, we have found that generically, the centers can be ordered for the Loewner order when $n= 3$ but not always for $n>3$. (Recall that the Loewner order is a partial order on the space of SPD matrices, according to which $A\geq B$ if $A-B$ is positive semi-definite.) This can be explained by the fact that the Riemannian metric \eqref{metric} sorts by rank, and therefore the covariance matrices are segmented into those of rank close to zero (since we have centered the velocity field), those of rank $1$ and those of full rank. In the second row of Figure \ref{fig:villes}, these clusters are respectively shown in green, blue and red. As could be expected, the first cluster corresponds to zones with either an isolated trajectory or parallel trajectories, the second to simple crossings or variations of speed in a given trajectory, and the third to zones with high density and crossings involving many trajectories. Naturally, the choice of the size $r$ of the kernel's support has a great influence on the clustering, and it should be adjusted according to the minimum distance authorized between two aircrafts in a zone considered as non conflictual. 

\begin{figure}
\subfloat{\includegraphics[width=0.3\textwidth]{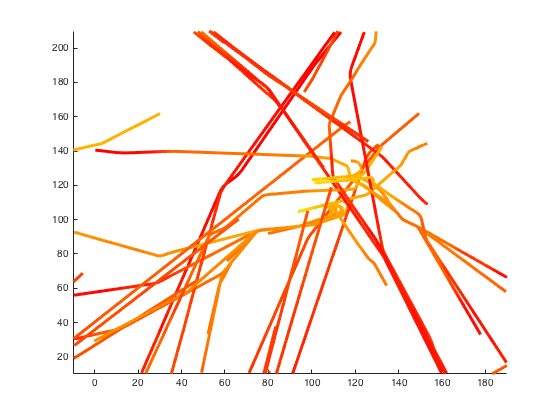}}
\subfloat{\includegraphics[width=0.3\textwidth]{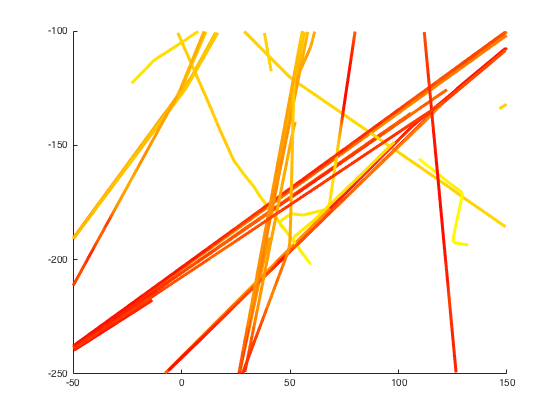}}
\subfloat{\includegraphics[width=0.3\textwidth]{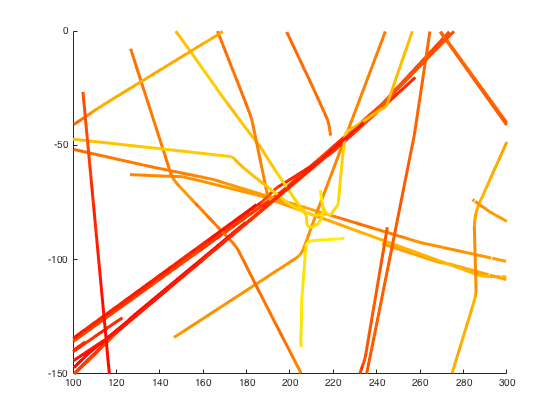}}\\
\subfloat{\includegraphics[width=0.3\textwidth]{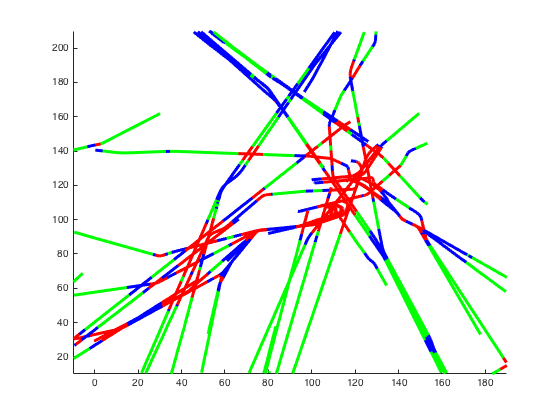}}
\subfloat{\includegraphics[width=0.3\textwidth]{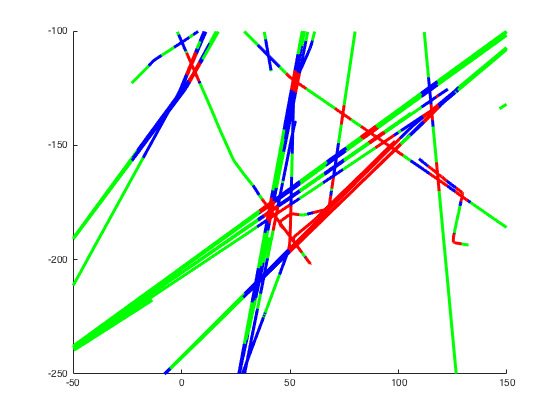}}
\subfloat{\includegraphics[width=0.3\textwidth]{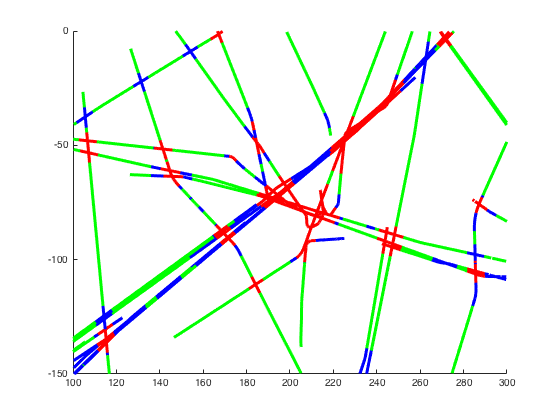}}
\caption{Clustering of the airspaces over Paris (left), Toulouse (middle) and Lyon (right).}
\label{fig:villes}
\end{figure}

Further results are shown in Figure \ref{fig:france}, where we consider the whole French airspace over different one-hour periods of time. The clusterings obtained using CLRQ is shown in the middle column. To illustrate the importance of the Riemannian setting with respect to the Euclidean one, we show results of Competitive Learning Vector Quantization on the same datasets, i.e. the same algorithm where the centers are updated using straight lines (linear interpolations between the matrix coefficients) and the distances are computed using the Frobenius norm. These results are shown in the right column of Figure \ref{fig:france}, and are visually less convincing: many zones that are perceived as complex by the human eye are not classified as such by the algorithm. In Figure \ref{fig:init}, we show that the initialization has little influence on the segmentation of the airspace, which is satisfactory.

\begin{figure}
\subfloat{\includegraphics[width=0.3\textwidth,height=0.3\textwidth]{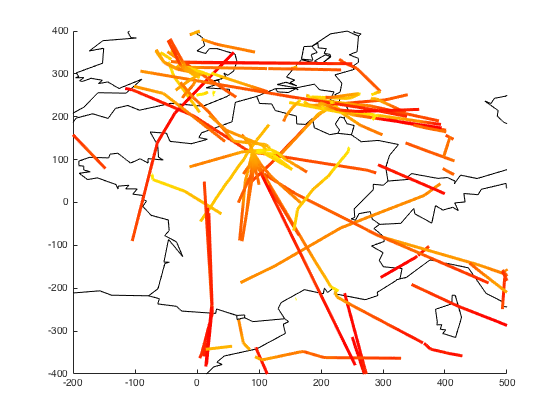}}
\subfloat{\includegraphics[width=0.3\textwidth,height=0.3\textwidth]{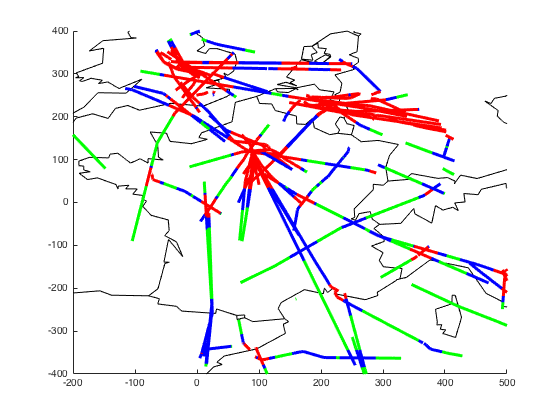}}
\subfloat{\includegraphics[width=0.3\textwidth,height=0.3\textwidth]{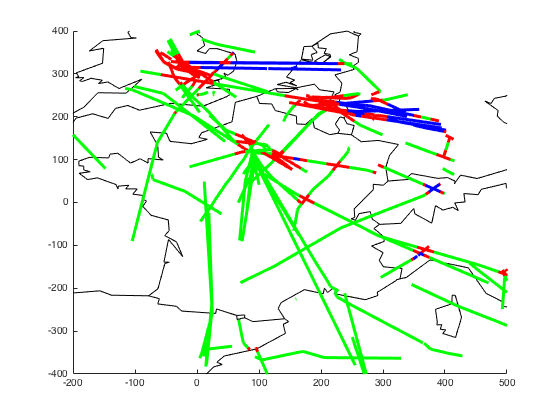}}\vspace*{-1em}\\
\subfloat{\includegraphics[width=0.3\textwidth,height=0.3\textwidth]{H1-v.png}}
\subfloat{\includegraphics[width=0.3\textwidth,height=0.3\textwidth]{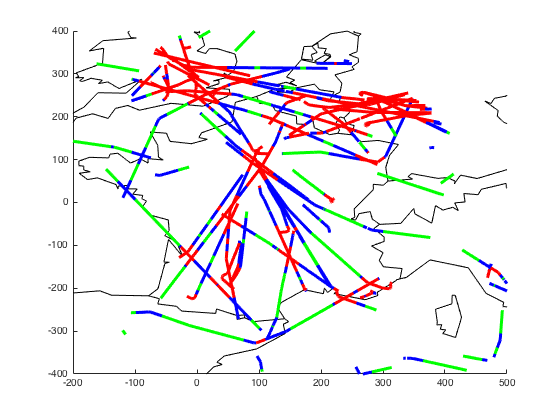}}
\subfloat{\includegraphics[width=0.3\textwidth,height=0.3\textwidth]{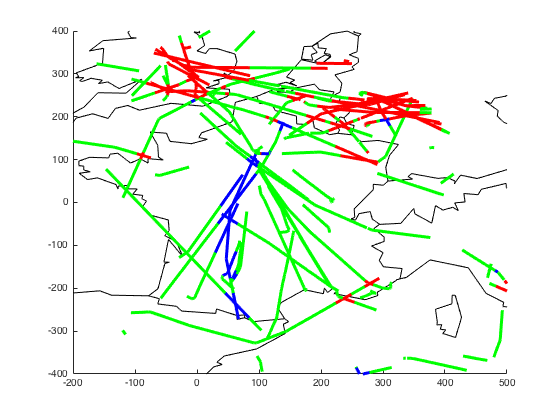}}\vspace*{-1em}\\
\subfloat{\includegraphics[width=0.3\textwidth,height=0.3\textwidth]{H3-v.png}}
\subfloat{\includegraphics[width=0.3\textwidth,height=0.3\textwidth]{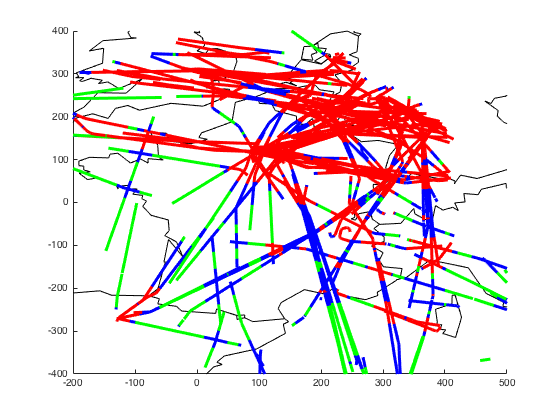}}
\subfloat{\includegraphics[width=0.3\textwidth,height=0.3\textwidth]{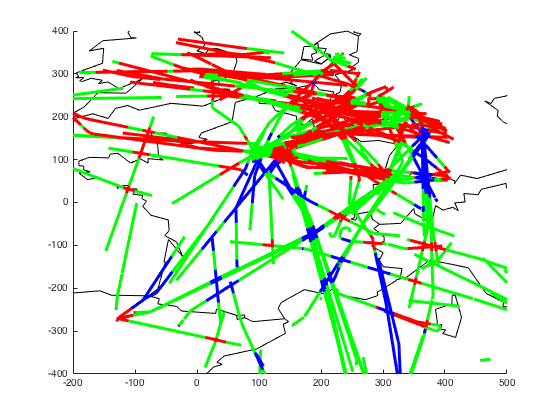}}\vspace*{-1em}\\
\subfloat{\includegraphics[width=0.3\textwidth,height=0.3\textwidth]{H6-v.png}}
\subfloat{\includegraphics[width=0.3\textwidth,height=0.3\textwidth]{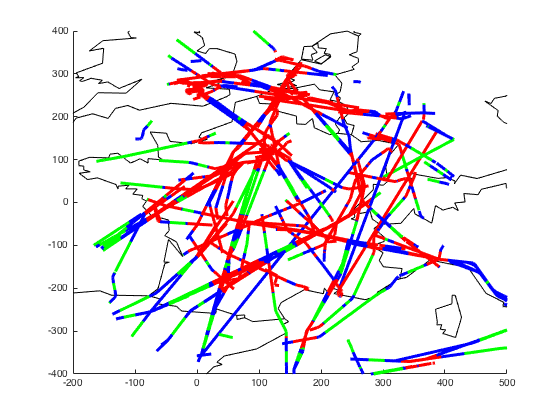}}
\subfloat{\includegraphics[width=0.3\textwidth,height=0.3\textwidth]{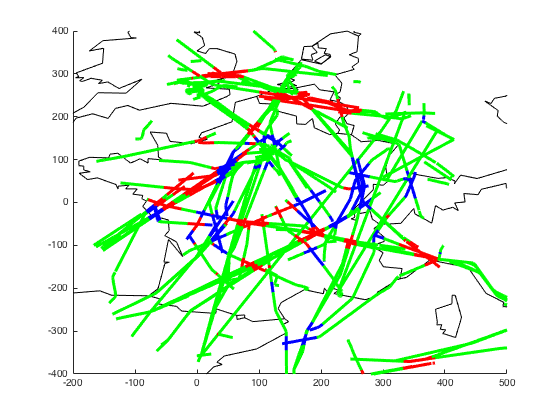}}
\caption{Clustering of the French airspace using Riemannian quantization (middle) versus vector quantization (right).}
\label{fig:france}
\end{figure}

\begin{figure}
\subfloat{\includegraphics[width=0.3\textwidth,height=0.3\textwidth]{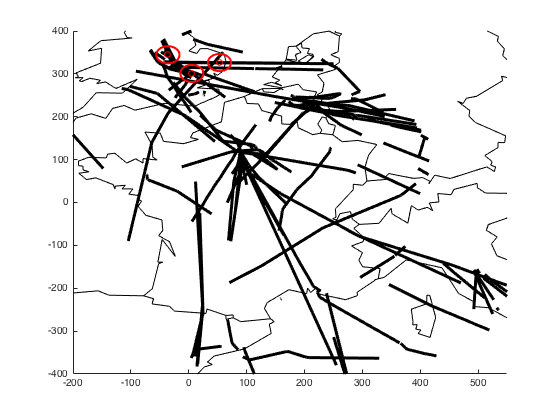}}
\subfloat{\includegraphics[width=0.3\textwidth,height=0.3\textwidth]{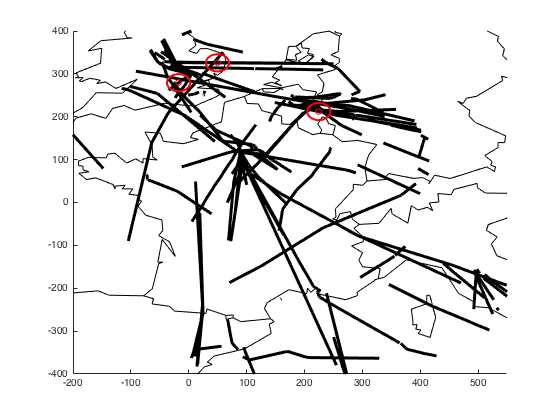}}
\subfloat{\includegraphics[width=0.3\textwidth,height=0.3\textwidth]{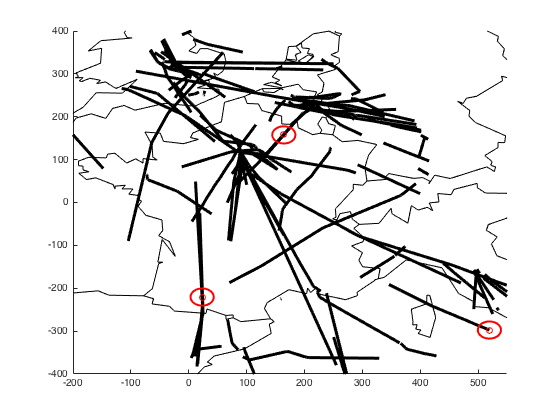}}\\
\subfloat{\includegraphics[width=0.3\textwidth,height=0.3\textwidth]{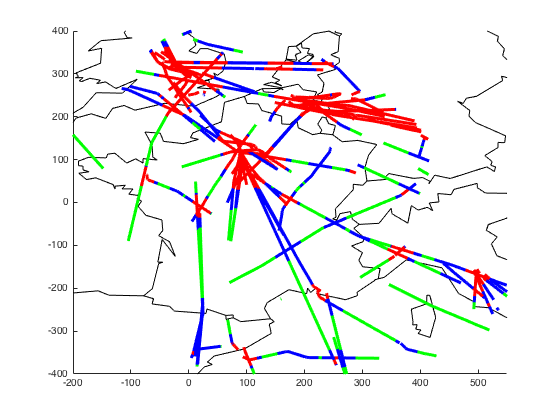}}
\subfloat{\includegraphics[width=0.3\textwidth,height=0.3\textwidth]{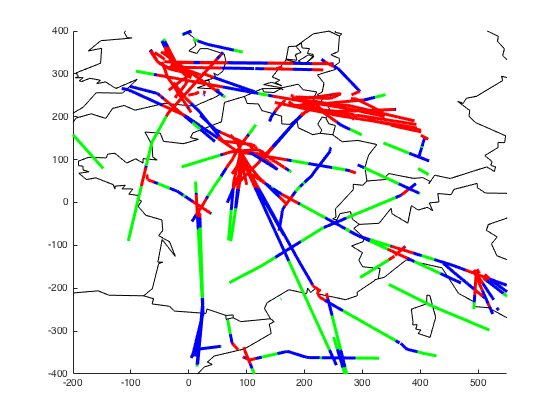}}
\subfloat{\includegraphics[width=0.3\textwidth,height=0.3\textwidth]{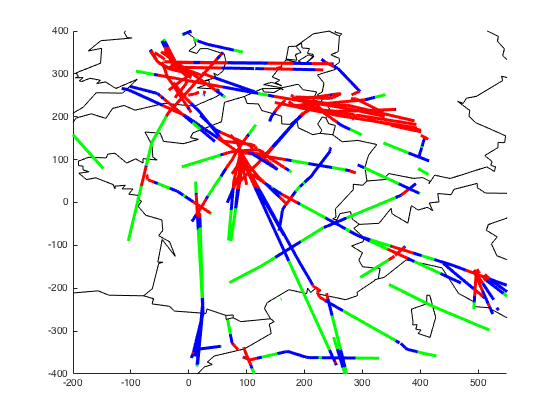}}
\caption{Clustering of the French airspace with 3 different initializations.}
\label{fig:init}
\end{figure}

\subsubsection{Comparing summaries}

Finally, we propose a way to compare our different summaries. A natural way to do so is through discrete optimal transport, which allows one to compute the distance between two discrete measures 
\begin{equation*}
\mu = \sum_{i=1}^m \mu_i\delta_{A_i} \quad \text{and} \quad \nu = \sum_{j=1}^n\nu_j\delta_{Bj}.
\end{equation*}
In our case, the $A_i$'s and $B_j$'s are SPD matrices. Optimal transport seeks to transport the mass from $\mu$ to $\nu$ in a way that minimizes a certain cost. Formally, a transport plan between $\mu$ and $\nu$ is a matrix $\pi = (\pi_{ij})_{i,j}$ with non-negative entries that verifies the two following properties for all $i=1,\hdots,m$ and $j=1,\hdots,n$,
\begin{equation*}
\sum_{j=1}^n \pi_{ij}=\mu_i \quad \text{and} \quad \sum_{i=1}^m \pi_{ij} = \nu_j.
\end{equation*}
The set of transport plans between $\mu$ and $\nu$ is denoted by $\Gamma(\mu,\nu)$. Intuitively, the value $\pi_{ij}$ represents the part of the mass $\mu_i$ transported from $A_i$ to $B_j$ to reconstruct $\nu_j$. Here, we measure the cost of transporting an infinitesimal unit of mass from $A_i$ to $B_j$ by the corresponding geodesic distance $d(A_i,B_j)$. The optimal transport plan is chosen to minimize the global cost, i.e. to be a solution of the following minimization problem
\begin{equation}\label{ot}
\min \left\{ \sum_{i=1}^m\sum_{j=1}^n \pi_{ij}d(A_i,B_j);\,\,\pi\in\Gamma(\mu,\nu)\right\}.
\end{equation}
Note that the minimal cost is the discrete $L^2$-Wasserstein distance between $\mu$ and $\nu$. In general, the linear programming problem \eqref{ot} is difficult and many different algorithms have been developed to solve it in various special cases \cite{merigot2016}. However, in our setting, the discrete measure involved are supported by a very small number of points ($3$ in the examples shown) and problem \eqref{ot} presents no difficulties. The distances between the summaries of the different traffic situations shown in the middle column of Figure \ref{fig:france} are given in Table \ref{table:france}. As expected, the first two situations are deemed similar but very different from the third one which has a much more complex traffic. The last situation is intermediary. In comparison, the different summaries of Figure \ref{fig:init} obtained for the same traffic situation but different initializations are at small distances from one another, as shown in Table \ref{table:init}.

\begin{table}
\caption{Distances between the summaries of Figure \ref{fig:france}.}
\label{table:france}
\begin{tabular}{|c|c|c|c|}
\hline
0.00 & 1.92 & 6.74 & 4.55 \\
\hline 
1.92 & 0.00 & 8.31 & 6.07\\
\hline
6.74 & 8.31 & 0.00 & 1.22\\
\hline
4.55 & 6.07 & 1.22 & 0.00\\
\hline
\end{tabular}
\end{table}

\begin{table}
\caption{Distances between the summaries of Figure \ref{fig:init}.}
\label{table:init}
\begin{tabular}{|c|c|c|}
\hline
0.000 & 0.033 & 0.031  \\
\hline 
0.033 & 0.000 & 0.016\\
\hline
0.031 & 0.016 & 0.000 \\
\hline
\end{tabular}
\end{table}

\section{Conclusion}

We have proposed a gradient descent type algorithm to find the best finite discrete approximation of a probability measure on a Riemannian manifold. This algorithm is adapted to large sets of data as it is online, and yields a clustering on top of a finite summary of the non-Euclidean data of interest. It is convergent when the manifold is complete and when its injectivity radius is uniformly bounded from below by a positive constant. We have used it to compute summaries of air traffic images in the form of finite numbers of covariance matrices representing different levels of local complexity, with associated weights corresponding to the occurrence of these levels of complexity in the images. In future work, we will consider best finite constrained quantization, i.e. restriction to finite approximations with equal weights, for the segmentation of the airspace in zones of equal complexity.

\section*{Acknowledgements}

We would like to thank Salem Said and Lionel Bombrun for sharing with us their code to generate Gaussian samples in the hyperbolic half-plane.

\section*{Appendix}

\begin{lemma} 
\label{lemgrad}
Let $x\in M$. The gradient of the function $f:a\mapsto d^2(x,a)$ is given by 
\begin{equation*}
\nabla_af  = -2 \log_ax = -2 \overrightarrow{ax}.
\end{equation*}
\end{lemma}

\begin{proof}
Let $a\in M$, $u\in T_a M$, and $(-\epsilon,\epsilon)\ni t\mapsto a(t)$ such that $a(0)=a, \dot a(0)=u$, so that the differential of $f$ at $a$ in $u$ is written, in terms of the norm associated to the Riemannian metric,
\begin{equation*}
T_af(u) = \left.\frac{d}{dt}\right|_{t=0} f(a(t))=\left.\frac{d}{dt}\right|_{t=0} d^2(x,a(t))=\left.\frac{d}{dt}\right|_{t=0}\left\|\overrightarrow{xa(t)}\right\|^2,
\end{equation*}
where $\overrightarrow{xa(t)}=\log_xa(t)$. Consider now $(-\epsilon,\epsilon)\times[0,1]\ni (t,s)\mapsto a(t,s)$ such that $a(t,0)=x, a(t,1)=a(t)$, and $s\mapsto a(t,s)$ is a geodesic for all $t$. Notice that $a_s(t,0)=\overrightarrow{xa(t)}$ and $a_t(0,1)=u$ if subscripts denote partial derivatives. Then, if $\nabla$ denotes covariant derivative
\begin{equation*}
T_af(u) = 2\left\langle \left.\nabla_{t}\right|_{t=0}\overrightarrow{xa(t)},\overrightarrow{xa(0)}\right\rangle=2\left\langle \nabla_{t}a_s(0,0),a_s(0,0)\right\rangle=2\left\langle \nabla_{s}J(0),\gamma_s(0)\right\rangle,
\end{equation*}
if $J(s):=a_t(0,s)$ and $\gamma(s):= a(0,s)$. Since it measures the way geodesics starting from $x$ spread out, the vector field $J$ is a Jacobi field along the geodesic $\gamma$ and therefore verifies the classical equation in terms of the curvature tensor $\mathcal R$ of $M$, $\nabla^2_sJ=-\mathcal R(J,\gamma_s)\gamma_s$, which implies $\langle \nabla_s^2J,\gamma_s\rangle=0$, leading to the scalar product $\left\langle \nabla_{s}J,\gamma_s\right\rangle$ being constant. This gives, since $J(0)=0$,
\begin{equation*}
\left\langle J(s),\gamma_s(s)\right\rangle =\left\langle \nabla_{s}J(0),\gamma_s(0)\right\rangle s+ \left\langle J(0),\gamma_s(0)\right\rangle= \left\langle \nabla_{s}J(0),\gamma_s(0)\right\rangle s,
\end{equation*}
and so, if $\overrightarrow{xa}^\parallel$ denotes the parallel transport of the vector $\overrightarrow{xa}=\log_xa$ from $x$ to $a$ along $\gamma$, then
\begin{equation*}
\left\langle \nabla_{s}J(0),\gamma_s(0)\right\rangle =\left\langle J(1),\gamma_s(1)\right\rangle =\langle u, \gamma_s(1) \rangle.
\end{equation*}
Setting $\tilde\gamma(s)=\gamma(1-s)$, we get $\gamma_s(1)=-\tilde\gamma_s(0)=-\overrightarrow{ax}=-\log_ax$ and so finally
\begin{equation*}
T_af(u)=-2\langle u, \log_ax \rangle,
\end{equation*}
which completes the proof.
\end{proof}

\begin{lemma}
\label{lemmass}
Let $\alpha=(a_1,\hdots,a_n)\in M^n$ and $C_i(\alpha)$ denote the Voronoi cell associated to $a_i$ for all $i=1,\hdots,n$. No mass is assigned by $\mu$ to the boundaries of the Voronoi diagram
\begin{equation*}
\mu\left(\partial C_i(\alpha)\right)=0.
\end{equation*}
\end{lemma}
\begin{proof}
For any $n$-tuple $\alpha=(a_1,\hdots,a_n)$, the $i^{th}$ Voronoi cell can be written
\begin{equation*}
C_i(\alpha) = \cap_{j\neq i} H(a_i,a_j), \quad \text{where} \quad H(a,b)=\{x\in K, d(x,a)\leq d(x,b)\},
\end{equation*}
which gives
\begin{equation*}
\partial C_i(\alpha) = C_i(\alpha)\cap \mathring{C_i}(\alpha)^c = C_i(\alpha)\cap \big(\ \cap_{j\neq i}\mathring{H}(a_i,a_j)\big)^c = \cup_{j\neq i} \partial H(a_i,a_j) \cap C_i(\alpha). 
\end{equation*}
Now for any $i\neq j$, the subset $\partial H(a_i,a_j)=\{x \in K, d(x,a_i)=d(x,a_j)\}$ defined as the kernel of the submersion $x \mapsto d(a_i,x)-d(a_j,x)$ is a submanifold of $M$ of codimension 1,  yielding $\mu(\partial H(a_i,a_j))=0$ and therefore $\mu(\partial C_i(\alpha))=0$.
\end{proof}

\bibliographystyle{plain}
\bibliography{bibliographie}

\begin{thebibliography}{10}

\bibitem{arnaudon2013riemannian}
M.~Arnaudon, F.~Barbaresco, and L.~Yang.
\newblock Riemannian medians and means with applications to radar signal
  processing.
\newblock {\em IEEE Journal of Selected Topics in Signal Processing},
  7(4):595--604, 2013.

\bibitem{arnaudon2013medians}
Marc Arnaudon, Fr{\'e}d{\'e}ric Barbaresco, and Le~Yang.
\newblock Medians and means in riemannian geometry: existence, uniqueness and
  computation.
\newblock In {\em Matrix Information Geometry}, pages 169--197. Springer, 2013.

\bibitem{biau2008}
G.~Biau, L.~Devroye, and G.~Lugosi.
\newblock On the performance of clustering in hilbert spaces.
\newblock {\em IEEE Transactions on Information Theory}, 54(2):781--790, 2008.

\bibitem{bigot2017}
J.~Bigot, R.~Gouet, T.~Klein, and A.~L{\'o}pez.
\newblock Geodesic pca in the wasserstein space by convex pca.
\newblock In {\em Annales de l'Institut Henri Poincar{\'e}, Probabilit{\'e}s et
  Statistiques}, volume~53, pages 1--26. Institut Henri Poincar{\'e}, 2017.

\bibitem{bonnabel2013}
S.~Bonnabel.
\newblock Stochastic gradient descent on riemannian manifolds.
\newblock {\em IEEE Transactions on Automatic Control}, 58(9):2217--2229, 2013.

\bibitem{bouchitte2011}
G.~Bouchitt{\'e}, C.~Jimenez, and R.~Mahadevan.
\newblock Asymptotic analysis of a class of optimal location problems.
\newblock {\em Journal de math{\'e}matiques pures et appliqu{\'e}es},
  95(4):382--419, 2011.

\bibitem{cabrelli1995}
C.~A. Cabrelli and U.~M. Molter.
\newblock The kantorovich metric for probability measures on the circle.
\newblock {\em Journal of Computational and Applied Mathematics},
  57(3):345--361, 1995.

\bibitem{calvo1991}
Miquel Calvo and Josep~Maria Oller.
\newblock An explicit solution of information geodesic equations for the
  multivariate normal model.
\newblock {\em Statistics \& Risk Modeling}, 9(1-2):119--138, 1991.

\bibitem{cheeger2008}
Jeff Cheeger and David~G Ebin.
\newblock {\em Comparison theorems in Riemannian geometry}, volume 365.
\newblock American Mathematical Soc., 2008.

\bibitem{COOK2015149}
A.~Cook, H.~Blom, F.~Lillo, R.~Mantegna, S.~Miccichè, D.~Rivas, R.~Vázquez,
  and M.~Zanin.
\newblock Applying complexity science to air traffic management.
\newblock {\em Journal of Air Transport Management}, 42:149 -- 158, 2015.

\bibitem{del2010}
D.~Delahaye and S.~Puechmorel.
\newblock Air traffic complexity based on dynamical systems.
\newblock In {\em Proceedings of the 49th CDC conference}. IEEE, 2010.

\bibitem{fletcher2007}
P.~T. Fletcher and S.~Joshi.
\newblock Riemannian geometry for the statistical analysis of diffusion tensor
  data.
\newblock {\em Signal Processing}, 87(2):250--262, 2007.

\bibitem{fletcher2004}
P.~T. Fletcher, C.~Lu, S.~M. Pizer, and S.~Joshi.
\newblock Principal geodesic analysis for the study of nonlinear statistics of
  shape.
\newblock {\em IEEE transactions on medical imaging}, 23(8):995--1005, 2004.

\bibitem{frechet1948}
Maurice Fr{\'e}chet.
\newblock Les {\'e}l{\'e}ments al{\'e}atoires de nature quelconque dans un
  espace distanci{\'e}.
\newblock {\em Ann. Inst. H. Poincar{\'e}}, 10(3):215--310, 1948.

\bibitem{graf2007}
S.~Graf and H.~Luschgy.
\newblock {\em Foundations of quantization for probability distributions}.
\newblock Springer, 2007.

\bibitem{iacobelli2016}
M.~Iacobelli.
\newblock Asymptotic quantization for probability measures on riemannian
  manifolds.
\newblock {\em ESAIM: Control, Optimisation and Calculus of Variations},
  22(3):770--785, 2016.

\bibitem{jost2008}
J.~Jost.
\newblock {\em Riemannian geometry and geometric analysis}, volume 42005.
\newblock Springer, 2008.

\bibitem{karcher1977}
H.~Karcher.
\newblock Riemannian center of mass and mollifier smoothing.
\newblock {\em Communications on pure and applied mathematics}, 30(5):509--541,
  1977.

\bibitem{kendall1984}
D.~G. Kendall.
\newblock Shape manifolds, procrustean metrics, and complex projective spaces.
\newblock {\em Bulletin of the London Mathematical Society}, 16(2):81--121,
  1984.

\bibitem{kloeckner2012}
B.~Kloeckner.
\newblock Approximation by finitely supported measures.
\newblock {\em ESAIM: Control, Optimisation and Calculus of Variations},
  18(2):343--359, 2012.

\bibitem{dyndens1998}
I.~V. Laudeman, S.~G. Shelden, R.~Branstrom, and C.~L. Brasil.
\newblock Dynamic density: An air traffic management metric.
\newblock Technical Report NASA/TM-1998-112226, NASA, 1998.

\bibitem{lebrigant2017}
A.~Le~Brigant.
\newblock Computing distances and geodesics between manifold-valued curves in
  the srv framework.
\newblock {\em arXiv preprint arXiv:1601.02358}, 2016.

\bibitem{lee2007}
K.~Lee, E.~Feron, and A.~Prichett.
\newblock Air traffic complexity : An input-output approach.
\newblock In {\em Proceedings of the US Europe ATM Seminar}, pages 2--9.
  Eurocontrol-FAA, 2007.

\bibitem{lovric2000}
Miroslav Lovri{\'c}, Maung Min-Oo, and Ernst~A Ruh.
\newblock Multivariate normal distributions parametrized as a riemannian
  symmetric space.
\newblock {\em Journal of Multivariate Analysis}, 74(1):36--48, 2000.

\bibitem{merigot2016}
Q.~M{\'e}rigot and E.~Oudet.
\newblock Discrete optimal transport: complexity, geometry and applications.
\newblock {\em Discrete \& Computational Geometry}, 55(2):263--283, 2016.

\bibitem{alldata2018}
G.~Mykoniatis, F.~Nicol, and S.~Puechmorel.
\newblock A new representation of air traffic data adapted to complexity
  assessment.
\newblock In {\em Proceedings of ALLDATA2018}, pages 28--33. IARA, 2018.

\bibitem{nadaraya1964}
E.~A. Nadaraya.
\newblock On estimating regression.
\newblock {\em Theory of Probability \& Its Applications}, 9(1):141--142, 1964.

\bibitem{pages2006}
G.~Pages.
\newblock Quadratic optimal functional quantization of stochastic processes and
  numerical applications.
\newblock In {\em Monte Carlo and Quasi-Monte Carlo Methods 2006}, pages
  101--142. Springer, 2008.

\bibitem{pages2015}
G.~Pag{\`e}s.
\newblock Introduction to vector quantization and its applications for
  numerics.
\newblock {\em ESAIM: proceedings and surveys}, 48:29--79, 2015.

\bibitem{pages2004}
G.~Pag{\`e}s, H.~Pham, and J.~Printems.
\newblock Optimal quantization methods and applications to numerical problems
  in finance.
\newblock In {\em Handbook of computational and numerical methods in finance},
  pages 253--297. Springer, 2004.

\bibitem{pennec2006intrinsic}
X.~Pennec.
\newblock Intrinsic statistics on riemannian manifolds: Basic tools for
  geometric measurements.
\newblock {\em Journal of Mathematical Imaging and Vision}, 25(1):127, 2006.

\bibitem{pennec2006riemannian}
X.~Pennec, P.~Fillard, and N.~Ayache.
\newblock A riemannian framework for tensor computing.
\newblock {\em International Journal of computer vision}, 66(1):41--66, 2006.

\bibitem{PRAN2011}
M.~Prandini, L.~Piroddi, S.~Puechmorel, and S.~L. Brazdilova.
\newblock Toward air traffic complexity assessment in new generation air
  traffic management systems.
\newblock {\em IEEE Transactions on Intelligent Transportation Systems},
  12(3):809--818, Sept 2011.

\bibitem{rabin2011}
J.~Rabin, J.~Delon, and Y.~Gousseau.
\newblock Transportation distances on the circle.
\newblock {\em Journal of Mathematical Imaging and Vision}, 41(1-2):147, 2011.

\bibitem{said2017}
S.~Said, L.~Bombrun, Y.~Berthoumieu, and J.~H. Manton.
\newblock Riemannian gaussian distributions on the space of symmetric positive
  definite matrices.
\newblock {\em IEEE Transactions on Information Theory}, 63(4):2153--2170,
  2017.

\bibitem{sommer2010}
S.~Sommer, F.~Lauze, S.~Hauberg, and M.~Nielsen.
\newblock Manifold valued statistics, exact principal geodesic analysis and the
  effect of linear approximations.
\newblock In {\em European conference on computer vision}, pages 43--56.
  Springer, 2010.

\bibitem{watson1964}
G.~S. Watson.
\newblock Smooth regression analysis.
\newblock {\em Sankhy{\=a}: The Indian Journal of Statistics, Series A}, pages
  359--372, 1964.

\bibitem{yin2010}
J.~Yin, Z.~Geng, R.~Li, and H.~Wang.
\newblock Nonparametric covariance model.
\newblock {\em Statistica Sinica}, 20:469, 2010.

\end{thebibliography}

\end{document}